\pdfoutput=1
\pdfpagewidth=\paperwidth
\pdfpageheight=\paperheight
\documentclass[12pt, letterpaper]{article}

\DeclareFixedFont{\MyTitleFont}{OT1}{\familydefault}{m}{n}{18pt}
\DeclareFixedFont{\MyAuthorFont}{OT1}{\familydefault}{m}{n}{13pt}
\DeclareFixedFont{\MyAbstractTitleFont}{OT1}{\familydefault}{m}{it}{12pt}
\DeclareFixedFont{\MyAbstractFont}{OT1}{\familydefault}{m}{it}{11pt}
\DeclareFixedFont{\MySubtitleFont}{OT1}{\familydefault}{m}{n}{15pt}
\DeclareFixedFont{\MySubSubtitleFont}{OT1}{\familydefault}{m}{n}{13pt}
\DeclareFixedFont{\MySubSubSubtitleFont}{OT1}{\familydefault}{m}{n}{11pt}
\DeclareFixedFont{\MyTextFont}{OT1}{\familydefault}{m}{n}{11pt}

\usepackage{geometry}
\usepackage{titling}

\title{\MyTitleFont Simplified Analysis on Filtering Sensitivity Trade-offs in Continuous- and Discrete-Time Systems   \vspace{0em}}

\author{\MyAuthorFont Neng Wan$^1$, Dapeng Li$^2$, Lin Song$^1$, and Naira Hovakimyan$^1$}
\date{}

\usepackage{titlesec}

\usepackage[marginal]{footmisc}
\usepackage{abstract}

\usepackage[sort]{cite}

\usepackage{indentfirst}
\usepackage{amsthm}
\usepackage{amsfonts}
\usepackage{amsmath}
\usepackage{booktabs}
\usepackage{bm}
\usepackage{subfigure}
\usepackage{enumerate}
\usepackage[mathcal]{euscript}
\usepackage{float}
\usepackage{graphics} 
\usepackage{epsfig} 
\usepackage[bookmarks]{hyperref}
\hypersetup{colorlinks = true, citecolor = red, linkcolor = blue, urlcolor = black}

\newtheorem{theorem}{Theorem}
\newtheorem{remark}{Remark}

\newtheorem{lemma}{Lemma}

\begin{document}

\maketitle


\footnotetext[0]{$^{1}$N. Wan, L. Song, and N. Hovakimyan are with the Department of Mechanical Science and Engineering and Coordinated Science Laboratory, University of Illinois at Urbana-Champaign, Urbana, IL 61801, USA (e-mail: \{nengwan2, linsong2, nhovakim\}@illinois.edu).}
\footnotetext[0]{$^{2}$D. Li is with the Anker Innovations, Shenzhen, Guangdong 518000, China (e-mail: dapeng.ustc@gmail.com).}
\vspace{-4em}

\begin{abstract}
{\vspace{0.5em}
	A simplified analysis is performed on the Bode-type filtering sensitivity trade-off integrals, which capture the sensitivity characteristics of the estimate and estimation error with respect to the process input and estimated signal in continuous- and discrete-time linear time-invariant filtering systems. Compared with the previous analyses based on complex analysis and Cauchy's residue theorem, the analysis results derived from the simplified method are more explicit, thorough, and require less restrictive assumptions. For continuous-time filtering systems, our simplified analysis reveals that apart from the non-minimum phase zero sets reported in the previous literature, the value and boundedness of filtering sensitivity integrals are also determined by the leading coefficients, relative degrees, minimum phase zeros, and poles of plants and filters. By invoking the simplified method, a comprehensive analysis on the discrete-time filtering sensitivity integrals is conducted for the first time. Numerical examples are provided to verify the validity and correctness of the simplified analysis.
}
\end{abstract}
\vspace{-0.5em}

\titleformat*{\section}{\centering\MySubtitleFont}
\titlespacing*{\section}{0em}{1.25em}{1.25em}[0em]

\section{Introduction}\label{sec1}
Filtering trade-offs or fundamental limitations of filtering, rooted in the physical limitations of filtering systems, capture the trade-offs between the pursuit of estimation performance and the constraints in filter design. Some identified filtering trade-offs include the filtering sensitivity integrals \cite{Goodwin_Auto_1995, Seron_CDC_1995, Braslavsky_CDC_1996, Goodwin_TAC_1997, Carrasco_Auto_2014}, lowest achievable estimation errors \cite{Braslavsky_Auto_1999, Li_Auto_2016, Wan_arxiv_2022}, and lower bounds on estimation error variance \cite{Fang_CDC_2017}. In this paper, we investigate the Bode-type filtering sensitivity integrals of continuous- and discrete-time linear time-invariant (LTI) systems via the simplified approach \cite{Wu_TAC_1992, Wan_TAC_2020}. In contrast to the existing studies, our simplified analysis provides more insights and explicit results, removes several restrictive assumptions on the continuous-time filtering sensitivity functions, and for the first time, presents the results of discrete-time filtering trade-off integrals.

Inspired by the sensitivity functions and trade-off integrals of LTI control systems \cite{Bode_1945, Freudenberg_1985, Sung_IJC_1988, Sung_IJC_1989, Middleton_1991}, Goodwin et al. \cite{Goodwin_Auto_1995, Goodwin_TAC_1997}, and Seron and Goodwin \cite{Seron_CDC_1995} put forward the filtering sensitivity functions and trade-off integrals for continuous-time LTI filtering systems, which measure the relative effects of process input and estimated signal on the estimate and estimation error at each frequency point and over the entire frequency range. Among these filtering sensitivity integrals, the Bode-type integrals \cite{Seron_2012}, integrals of logarithmic sensitivity functions on the entire frequency range, are candid in describing the sensitivity characteristics of filtering systems in frequency domain, but are not always bounded in continuous-time setting even with the aid of weighting function $1/\omega^2$. On the other hand, although the Poisson-type sensitivity integrals \cite{Goodwin_TAC_1997}, integrals of logarithmic sensitivity functions weighted by the Poisson kernel, are more tractable in boundedness, the physical meaning and instructions delivered by the Poisson integrals are not as straightforward as the Bode-type integrals. Meanwhile, similar to many early investigations of the control trade-off integrals \cite{Freudenberg_1985, Sung_IJC_1988, Sung_IJC_1989, Middleton_1991, Seron_2012}, all existing analyses of continuous-time filtering sensitivity integrals are based on complex analysis and Cauchy's residue theorem \cite{Braslavsky_CDC_1996, Goodwin_TAC_1997, Seron_2012}, which not only imposes multiple constraints on the analyticity, relative degree, and minimum phaseness of transfer and sensitivity functions, but conceals some factors that determine the values and boundedness of trade-off integrals in their final expressions. These limitations hindered our understanding and applications of Bode-type sensitivity integrals until the introduction of simplified method~\cite{Wu_TAC_1992}, which built on the elementary mathematical tools, shed more light on the analysis and interpretation of Bode-type sensitivity integrals. By resorting to this simplified method, Wu and Jonckheere \cite{Wu_TAC_1992} and Wan et al. \cite{Wan_TAC_2020} respectively studied the Bode integrals and complementary sensitivity Bode integrals in continuous- and discrete-time LTI control systems. Compared with the original studies in \cite{Freudenberg_1985, Sung_IJC_1988, Sung_IJC_1989, Middleton_1991}, their simplified analyses with milder assumptions are more comprehensive, and the results are more explicit and general in application.

Relying on the simplified method introduced in \cite{Wu_TAC_1992}, in this paper, we perform a simplified analysis on the Bode-type filtering sensitivity integrals of both continuous- and discrete-time LTI filtering systems. A general filtering configuration is established, and transfer functions of plants and filters are formulated under the minimum assumptions. Sensitivity functions and trade-off integrals are defined to characterize the relative impacts of process input or estimated signal on the estimate and estimation error. No restriction is imposed on the analyticity, relative degree, or minimum phaseness of transfer and sensitivity functions. For both continuous- and discrete-time filtering systems, we use the simplified method to derive the values and boundedness conditions of the trade-off integrals subject to different combinations of plants and filters, and show that apart from the non-minimum phase (NMP) zero sets reported in \cite{ Seron_2012}, the values and boundedness of Bode-type filtering trade-off integrals are also impacted by the relative degrees, leading coefficients, minimum phase (MP) zeros, and poles of plants and filters. Illustrative examples are provided at the end to examine and verify the theoretical analysis of this simplified method.

The remainder of this paper is organized as follows. \hyperref[sec2]{Section~II} introduces the preliminaries of the simplified method and formulates the filtering problems. \hyperref[sec3]{Sections~III} and \hyperref[sec4]{IV} respectively conduct the simplified analyses in continuous- and discrete-time filtering systems. \hyperref[sec5]{Section~V} presents the illustrative examples, and  \hyperref[sec6]{Section~VI} draws the conclusions. 

\vspace{0.5em}

\noindent \textit{Notations}: Throughout this paper, $\ln(\cdot)$ and $\log(\cdot)$ respectively denote the logarithms in base ${\rm e}$ and $2$. $\mathbb{C}^+$ represents the open right half-complex plane (ORHP). For a complex number $c$, $|c|$ stands for the modulus, and $|{\rm Re} \ c|$ denotes the absolute value of real part. Complex variables are denoted as $s = j\omega$ and $z = {\rm e}^{j\omega}$.

\section{Preliminaries and Problem Formulation}\label{sec2}
To formulate the filtering problem, we follow the setup in \cite{Goodwin_TAC_1997, Seron_2012} and consider the general filtering system as follows. 
\begin{figure}[H]
	\centerline{\includegraphics[width=0.65\columnwidth]{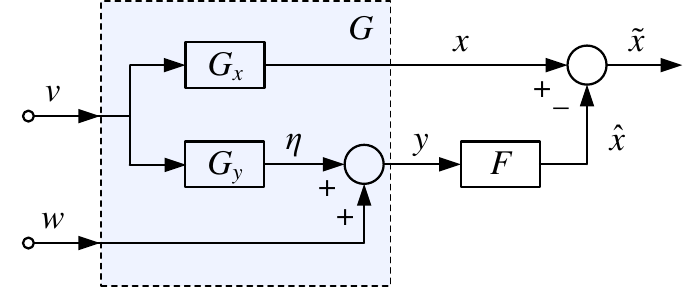}}
	\caption{General filtering configuration.}
	\label{fig1}
\end{figure}
\noindent Among the signals in \hyperref[fig1]{Fig.~1}, $v$ and $w$ are respectively the process and measurement inputs or noise; $\eta$ denotes the noise-free output; $y$ stands for the measured output; $x$ is the actual signal to be estimated, e.g., internal states $\xi$ of plant $G$ or noise-free output $\eta$; $\hat{x}$ is the estimate, and $\tilde{x} = x - \hat{x}$ is the estimation error. The plant model $G$ in \hyperref[fig1]{Fig.~1} is decomposed into two parts, $G_x$ and $G_y$, where $G_x$ denotes the transfer function from process input $v$ to estimated signal $x$, and $G_y$ represents the transfer function from process input $v$ to noise-free output $\eta$. Filtering process $F$ computes the estimate $\hat{x}$ from measured output $y$ and is assumed to be a bounded error estimator\footnote{A stable filter is a BEE if, for all initial states of the plant and filter, the estimation error $\tilde{x}$ is bounded whenever the inputs $v$ and $w$ are bounded.} (BEE), which is the minimum assumption that can be required from a filter \cite{Goodwin_TAC_1997}. For conciseness, in this paper, we primarily consider the single-input and single-output (SISO) systems, namely $v, w, x$ and $y \in \mathbb{R}$ in \hyperref[fig1]{Fig.~1}, and the multi-input and multi-output (MIMO) systems can be studied by rewriting our analysis into a matrix-vector form and following the MIMO framework presented in \cite{Braslavsky_CDC_1996, Seron_2012}. A state-space realization of the plant $G$ in \hyperref[fig1]{Fig.~1} is
\begin{equation}\label{state_space}
	\left[\begin{matrix}
		\dot{\xi}\\
		x\\
		y
	\end{matrix}\right] = \left[\begin{matrix}
		A & B & 0\\
		C_1 & D_1 & 0\\
		C_2 & D_2 & I\\
	\end{matrix}\right]
	\left[\begin{matrix}
		\xi\\
		v\\
		w
	\end{matrix}\right]
\end{equation}
where $\xi$ represents the internal states of plant $G$; the pair $(A, C_2)$ is assumed to be detectable, and $(A, B)$ is assumed to be stabilizable.

To quantify the sensitivity characteristics of the system in \hyperref[fig1]{Fig.~1}, we consider the filtering sensitivity functions \cite{Seron_2012}:
\begin{subequations}\label{pm_funs}
	\begin{align}
		P & = (G_x - FG_y) G_x^{-1} \label{p_fun}\\
		M & = FG_yG_x^{-1}, \label{m_fun}
	\end{align}
\end{subequations}
where the filtering sensitivity function $P$, which is the transfer function from the signal $x$ to the estimation error $\tilde{x}$, reflects the impact of the process input $v$ or estimated signal $x$ on the estimation error $\tilde{x}$, and the filtering sensitivity function $M$, which is the transfer function from the signal $x$ to the estimate $\hat{x}$, captures the relative effect of the process input $v$ or the estimated signal $x$ on the estimate $\hat{x}$. The integrals of logarithmic $|P|$ and $|M|$ over the frequency range, defined as the Bode-type filtering sensitivity integrals, are utilized as the trade-off metrics characterizing the sensitivity properties of the filtering systems in the frequency domain. Similar to the sensitivity function $S$ and complementary sensitivity function $T$ of LTI control systems \cite{Seron_2012, Wan_SCL_2019}, the pair of filtering sensitivity functions defined in \eqref{pm_funs} also satisfy the complementary condition 
\begin{equation}\label{comp_constraint}
	P + M = I
\end{equation}
at any complex frequency $s$ or $z$ that is not a pole of $P$ and $M$. To guarantee the well-definiteness of the filtering problem and the properness of the sensitivity functions in \eqref{pm_funs}, we impose the following  assumptions: A1)\label{ass1} $G_x$ is right-invertible; A2)\label{ass2} $F$ is proper and stable; A3)\label{ass3} $FG_yG_x^{-1}$ is proper, or equivalently, $G_yG_x^{-1}$ is proper, and A4)\label{ass4} $G_x - FG_y$ is stable.
\begin{remark}\label{rem1}
	\hyperref[ass1]{A1)} guarantees the well-posedness of the sensitivity function $P$ in \eqref{p_fun}. When the filtering system in \hyperref[fig1]{Fig.~1} is SISO, \hyperref[ass1]{A1)} always holds, while for the MIMO filtering systems, \hyperref[ass1]{A1)} requires that there exist at least as many process inputs as the signals to be estimated. \hyperref[ass2]{A2)} and \hyperref[ass3]{A3)} ensure the properness of transfer and sensitivity functions in \hyperref[fig1]{Fig.~1} and \eqref{pm_funs}. \hyperref[ass4]{A4)} is a necessary and sufficient condition for the filter $F$ to be a BEE \cite{Seron_2012}, which has two major implications --- i) unstable poles shared by $G_x$ and $G_y$ must be canceled by the NMP zeros of $G_x - FG_y$, and ii) unstable poles of $G_y$ that do not appear in $G_x$ must be canceled by the NMP zeros of $F$. Compared with the previous studies of control and filtering trade-off integrals, no restriction is laid on the analyticity, relative degree, or minimum phaseness of transfer and sensitivity functions.
\end{remark}
\noindent In the following subsections, the transfer functions, sensitivity functions, and trade-off integrals are formulated in the continuous- and discrete-time filtering systems subject to \hyperref[ass1]{A1)}-\hyperref[ass4]{A4)}, and some preliminaries of the Bode-type filtering sensitivity integrals and simplified method are respectively introduced.

\titleformat*{\subsection}{\MySubSubtitleFont}
\titlespacing*{\subsection}{0em}{0.75em}{0.75em}[0em]
\subsection{Continuous-Time Filtering Systems}
We first consider the continuous-time plant model $G_x(s) = K_x  \prod_{i=1}^{m_x}(s - z_i) / \prod_{i=1}^{n_x}(s-p_i)$, where $K_x \in \mathbb{R}$ is the leading coefficient; $z_i$ and $p_i$ are respectively the $i$-th zero and pole of $G_x(s)$, and $m_x \leq n_x$ such that $G_x(s)$ is proper. For continuous-time systems, we use $\bar{z}_i$ to denote the NMP zeros with positive real parts in $\{z_i\}_{i=1}^{m_x}$, and $\bar{p}_i$ to indicate the unstable poles with positive real parts in $\{ p_i \}_{i=1}^{n_x}$. This rule also applies to the zeros and poles with subscripts other than~$i$. The transfer function $G_x(s)$ can then be rewritten as
\begin{equation}\label{TF_Gx}
	G_x(s) = \frac{K_x \cdot \prod_{i=1}^{\bar{m}_x}(s - \bar{z}^{}_i)  \prod_{i=1}^{m_x - \bar{m}_x}(s - z_i)}{\prod_{i=1}^{\bar{n}_x}(s-\bar{p}^{}_{i})  \prod_{i=1}^{n_x - \bar{n}_x}(s-p^{}_{i})},
\end{equation}
where $\bar{m}_x$ and $\bar{n}_x$ respectively denote the amounts of NMP zeros and unstable poles in $G_x(s)$, and we can readily imply that $\prod_{i=1}^{m_x}(s - z_i) = \prod_{i=1}^{\bar{m}_x}(s - \bar{z}^{}_i) \cdot \prod_{i=1}^{m_x - \bar{m}_x}(s - z_i)$, and  $\prod_{i=1}^{n_x}(s-p_i) = \prod_{i=1}^{\bar{n}_x}(s-\bar{p}^{}_{i}) \cdot \prod_{i=1}^{n_x - \bar{n}_x}(s-p^{}_{i})$.

The plant model $G_y$ in \hyperref[fig1]{Fig.~1} is described by the continuous-time transfer function $G_y(s) = K_y \prod_{a=1}^{m_y}(s - z_a) / \prod_{a=1}^{n_y}(s - p_a)$, where $K_y \in \mathbb{R}$ denotes the leading coefficient; $\{z_a\}_{a=1}^{m_y}$ and $\{p_a\}_{a=1}^{n_y}$ respectively stand for the zeros and poles of $G_y(s)$, and $m_y \leq n_y$ such that $G_y(s)$ is proper. Considering the fact that $G_y(s)$ contains all the unstable poles of $G(s)$ and thus $G_x(s)$ on account of the stabilizability and detectability assumptions in \eqref{state_space}, we can then reformulate $G_y(s)$ as follows
\begin{equation}\label{TF_Gy}
	G_y(s) = \frac{K_y \cdot \prod_{a=1}^{m_y}(s - z_a)}{\prod_{i=1}^{\bar{n}_{x}}(s-\bar{p}^{}_{i}) \prod_{a=1}^{\bar{n}_y}(s - \bar{p}^{}_{a}) \prod_{a=1}^{n_y - \bar{n}_y-\bar{n}_x}(s-p_{a})},
\end{equation} 
where $\{ \bar{p}_a^{}\}_{a=1}^{\bar{n}_y^{}}$ denote the unstable poles that only exist in $G_y(s)$ and do not appear in $G_x(s)$, and we can imply that $\prod_{a=1}^{n_y}(s - p_a) = \prod_{i=1}^{\bar{n}_{x}}(s-\bar{p}^{}_{i}) \prod_{a=1}^{\bar{n}_y}(s - \bar{p}^{}_{a}) \prod_{a=1}^{n_y - \bar{n}_y-\bar{n}_x}(s-p_{a})$. The filter $F$ in \hyperref[fig1]{Fig.~1} is realized by the transfer function $F(s) = K^{}_f  \prod_{b=1}^{m_f}(s - z^{}_b) / \prod_{b=1}^{n_f}(z - p^{}_b)$, where the leading coefficient $K_f \in \mathbb{R}$; $\{z^{}_b\}_{b=1}^{m_f}$ and $\{p^{}_b\}_{b=1}^{n_f}$ respectively denote the zeros and poles of $F(s)$, and $m_f \leq n_f$ such that $F(s)$ is proper. Since the unstable poles $\{\bar{p}^{}_{a}\}_{a=1}^{\bar{n}_y}$ solely in $G_y(s)$ must be canceled by the NMP zeros of $F(s)$ by \hyperref[ass4]{A4)} and \hyperref[rem1]{Remark~1}, we can then rewrite $F(s)$ as
\begin{equation}\label{TF_F}
	F(s) = \frac{K^{}_f \cdot  \prod_{a=1}^{\bar{n}_y}(s - \bar{p}_{a}^{}) \prod_{b=1}^{m^{}_f - \bar{n}_y}(s - z_b^{})}{\prod_{b=1}^{n_f}(s - p^{}_b)},
\end{equation}
where $\{p_b^{}\}_{b=1}^{n_f}$ are all stable poles by \hyperref[ass2]{A2)}, and $\prod_{b=1}^{m_f}(s - z^{}_b) = \prod_{a=1}^{\bar{n}_y}(s - \bar{p}_{a}^{}) \prod_{b=1}^{m^{}_f - \bar{n}_y}(s - z_b^{})$. Since $G_y(s)$ and $F(s)$ always appear as a product $F(s)G_y(s)$ in \eqref{pm_funs}, substituting \eqref{TF_Gy} and \eqref{TF_F} into the product and canceling the common factors, we have
\begin{equation}\label{TF_FG}
	F(s) G_y(s) = \frac{K \cdot \prod_{j=1}^{m}(s - z^{}_j)}{\prod_{i=1}^{\bar{n}_x}(s - \bar{p}_{i}^{}) \prod_{j=1}^{n - \bar{n}_x}(s - p_j^{})},
\end{equation}
where $K = K_f  K_y$, $\prod_{j=1}^m(s - z_j) = \prod_{a=1}^{m_y}(s -  z^{}_a)   \prod_{b=1}^{m_f -  \bar{n}_y}  (s  -  z^{}_b)$, and $\prod_{j=1}^{n- \bar{n}_x}(s - p^{}_j) = \prod_{a=1}^{n_y - \bar{n}_x - \bar{n}_y}(s - p^{}_a)  \prod_{b = 1}^{n_f} (s - p_b^{})$. By resorting to  Cauchy's residue theorem, researchers have derived the following Bode-type integral constraints for the sensitivity functions $P(s)$ and $M(s)$ in continuous-time LTI filtering systems \cite{Seron_2012}. 
\begin{lemma}\label{lem1}
	For the continuous-time filtering systems in \hyperref[fig1]{Fig.~1}, suppose that $P(s)$ and $M(s)$ in \eqref{pm_funs} are proper, and $F(s)$ in \eqref{TF_F} is a BEE. Let $\mathcal{Z}_P = \{ \xi_i \in \mathbb{C}^+: P(\xi_i) = 0 \}$ be the set of NMP zeros of $P(s)$, and $\mathcal{Z}_x = \{ \zeta_i \in \mathbb{C}^+:  G_x(\zeta_i) = 0 \ \textrm{and } F(\zeta_i)G_y(\zeta_i) \neq 0 \}$. If $P(j\infty) \neq 0$, the Bode-type integral of sensitivity function $P(s)$ then satisfies
	\begin{equation}\label{lem1_eq1}
			\hspace{-5pt}\frac{1}{2\pi}\int_{-\infty}^{\infty} \ln \left| \frac{P(j\omega)}{P(j \infty)}  \right|d\omega = \lim\limits_{s \rightarrow \infty} \frac{s[P(s) - P(\infty)]}{2P(\infty)} + \sum_{\xi_i \in \mathcal{Z}_P} \xi^{}_i - \sum_{\zeta_i \in \mathcal{Z}_x} \zeta^{}_i.
	\end{equation}
	Let $\mathcal{Z}_M = \{\delta_i \in \mathbb{C}^{+}: M(\delta_i) = 0\}$ be the set of NMP zeros of $M(s)$. If $M(0) \neq 0$, the Bode-type integral of the sensitivity function $M(s)$ then satisfies
	\begin{equation}\label{lem1_eq2}
			\frac{1}{2\pi}\int_{-\infty}^{\infty} \ln \left|  \frac{M(j\omega)}{M(0)}  \right|\frac{d\omega}{\omega^2}  = \frac{1}{2M(0)} \lim_{s\rightarrow 0}\frac{dM(s)}{ds}  + \sum_{\delta_i \in \mathcal{Z}_M} \delta^{-1}_i - \sum_{\zeta_i \in \mathcal{Z}_x} \zeta^{-1}_i.
	\end{equation}
\end{lemma}
\begin{remark}
	\hyperref[lem1]{Lemma~1} shows that sensitivity characteristics or trade-off integrals of continuous-time filtering systems are mainly impacted by the NMP zero sets $\mathcal{Z}_P$, $\mathcal{Z}_M$ and $\mathcal{Z}_x$. However, in practice, the trade-off integrals of many filtering systems with finite and non-zero NMP zero sets are unbounded, which can hardly be explained by \hyperref[lem1]{Lemma~1} and thus hinders our understanding and applications of the trade-off integrals. By using the simplified method, we can derive more explicit results and disclose the hidden factors that determine the values and boundedness of $P(j\infty)$, $M(0)$, limits, and hence the trade-off integrals in \eqref{lem1_eq1} and \eqref{lem1_eq2}.
\end{remark}
\noindent The following lemma, whose proof can be found in \cite{Wu_TAC_1992}, is the theoretical foundation of the simplified analysis in continuous-time systems. 
\begin{lemma}\label{lem2}
	For complex numbers $\alpha$ and $\beta$, we have
	\begin{equation}
		\int_{-\infty}^{\infty} \ln\left| \frac{j \omega - \alpha}{j \omega - \beta}  \right|^2 d\omega = 2\pi \cdot \left( |{\rm Re} \ \alpha|- |{\rm Re} \ \beta|   \right).
	\end{equation}
\end{lemma}

\subsection{Discrete-Time Filtering Systems}
Most of the preceding notations and interpretations for continuous-time filtering systems can be readily carried over to the discrete-time systems. In discrete-time scenario, we realize the plant $G_x$ in \hyperref[fig1]{Fig.~1} by the discrete-time transfer function $G_x(z) = K_x \prod_{i=1}^{m_x}(z - z_i) / \prod_{i=1}^{n_x}(z - p_i)$, where $K_x \in \mathbb{R}$, and $m_x \leq n_x$. Let $\bar{z}_i$ and $\bar{p}_i$ respectively denote the NMP zeros and unstable poles lying outside the unit disk, i.e., $|\bar{z}_i| > 1$ and $|\bar{p}_i| > 1$. We can then rewrite $G_x(z)$ as
\begin{equation}\label{DTF_Gx}
	G_x(z) = \frac{ K_x \cdot \prod_{i=1}^{\bar{m}_x}(z - \bar{z}_i) \prod_{i=1}^{m_x - \bar{m}_x}(z - z_i)}{\prod_{i=1}^{\bar{n}_x} (z - \bar{p}_i) \prod_{i=1}^{n_x - \bar{n}_x}(z - p_i)},
\end{equation}
where $\bar{m}_x$ and $\bar{n}_x$ are the amounts of NMP zeros and unstable poles, respectively; $\prod_{i=1}^{m_x}(z - z_i) = \prod_{i=1}^{\bar{m}_x}(z - \bar{z}_i)    \prod_{i=1}^{m_x - \bar{m}_x}(z  -  z_i)$, and $\prod_{i=1}^{n_x}(z - p_i) = \prod_{i=1}^{\bar{n}_x} (z - \bar{p}_i)   \prod_{i=1}^{n_x - \bar{n}_x}(z - p_i)$. By the stabilizability and detectability of \eqref{state_space} and the same argument for \eqref{TF_Gy}, the transfer function of plant $G_y(z) = K_y \prod_{a=1}^{m_y}(z - z_a) / \prod_{a=1}^{n_y}(z - p_a)$ with $K_y\in \mathbb{R}$ and $m_y \leq n_y$ can be reformulated as
\begin{equation}\label{DTF_Gy}
	G_y(z) = \frac{K_y \cdot \prod_{a=1}^{m_y}(z - z_a)}{\prod_{i=1}^{\bar{n}_x}(z - \bar{p}_i) \prod_{a=1}^{\bar{n}_y}(z - \bar{p}_a)  \prod_{a=1}^{n_y - \bar{n}_x - \bar{n}_y}(z - p_a)  }
\end{equation}
where $\{\bar{p}^{}_a\}_{a=1}^{\bar{n}_y}$ denote the unstable poles that exclusively exist in $G_y(z)$ and do not appear in $G_x(z)$, and $\prod_{a=1}^{n_y}(z - p_a) = \prod_{i=1}^{\bar{n}_x}(z - \bar{p}_i) \prod_{a=1}^{\bar{n}_y}(z - \bar{p}_a)  \prod_{a=1}^{n_y - \bar{n}_x - \bar{n}_y}(z - p_a)$. We realize the filter $F$ in \hyperref[fig1]{Fig.~1} with the discrete-time transfer function $F(z) = K_f \prod_{b=1}^{m_f}(z - z^{}_b) / \prod_{b=1}^{n_f}(z - p^{}_b)$, where $K_f \in \mathbb{R}$, and $m_f \leq n_f$. By \hyperref[ass4]{A4)} and the same reasoning for \eqref{TF_F}, we can reformulate $F(z)$ as
\begin{equation}\label{DTF_F}
	F(z) = \frac{K_f \cdot \prod_{a=1}^{\bar{n}_y}(z - \bar{p}^{}_a) \prod_{b=1}^{m_f - \bar{n}_y}(z- z^{}_b) }{\prod_{b=1}^{n_f}(z - p^{}_b)},
\end{equation}
where $\{p^{}_b\}_{b=1}^{n_f}$ are all stable poles inside the unit disk by \hyperref[ass2]{A2)}, and $\prod_{b=1}^{m_f}(z - z_b) = \prod_{a=1}^{\bar{n}_y}(z - \bar{p}_a) \prod_{b=1}^{m_f - \bar{n}_y}(z- z_b)$. Substituting \eqref{DTF_Gy} and \eqref{DTF_F} into the product $F(z)G_y(z)$ and canceling the common factors, we have
\begin{equation}\label{DTF_FG}
	F(z) G_y(z) = \frac{K \cdot \prod_{j=1}^{m}(z - z_j)}{ \prod_{i=1}^{\bar{n}_x}(z-\bar{p}_i) \prod_{j=1}^{n - \bar{n}_x}(z - p_j) },
\end{equation}
where the leading coefficient $K = K_f K_y$; $\prod_{j=1}^{n - \bar{n}_x}( z - p_j ) =  \prod_{a=1}^{n_y - \bar{n}_x - \bar{n}_y}(z - p_a) \prod_{b=1}^{n_f}(z-p^{}_b)$, and $\prod_{j=1}^{m}(z - z_j) = \prod_{j=1}^{\bar{m}}(z - \bar{z}_j)\prod_{j=1}^{m - \bar{m}}(z - z_j) = \prod_{a=1}^{m_y}(z - z_a) \prod_{b=1}^{m_f - \bar{n}_y} (z - z_b)$ with $\{\bar{z}_j\}_{j=1}^{\bar{m}}$ being the NMP zeros of $F(z)G_y(z)$. There is no existing result on the sensitivity trade-off integrals of discrete-time filtering systems. The following lemma, whose proof is available in \cite{Wu_TAC_1992}, gives the theoretical foundation of the simplified method in discrete-time systems. 
\begin{lemma}\label{lem3}
	For a complex number $\alpha$, we have
	\begin{equation*}
		\int_{-\pi}^{\pi} \log\left| e^{j\omega} - \alpha \right|^2 d\omega = 
		\begin{cases}
			2\pi \cdot \log |\alpha|^2, \qquad &{\rm if} \ |\alpha| > 1;\\
			0, \qquad &{\rm if} \ |\alpha| \leq 1;
		\end{cases}
	\end{equation*}
	where $|\alpha|$ is the modulus of $\alpha$.
\end{lemma}

\section{Continuous-Time Filtering Trade-offs}\label{sec3}
	In this section, a simplified analysis is applied to the Bode-type integrals of sensitivity functions $P(s)$ and $M(s)$ in \eqref{pm_funs}, respectively, which are also referred to as the continuous-time filtering $P$- and $M$-integrals in the subsequent context. 

\subsection{Continuous-Time Filtering $P$-Integral}
The simplified analysis of continuous-time filtering $P$-integral gives the following result. 
\begin{theorem}\label{thm1}
	For the continuous-time LTI filtering systems subject to plant models \eqref{TF_Gx}, \eqref{TF_Gy} and filter \eqref{TF_F}, as \hyperref[fig1]{Fig.~1} shows, the Bode-type integral of the filtering sensitivity function $P(s)$ defined in \eqref{p_fun} satisfies
	\begin{align}\label{thm1_eq0}
		\frac{1}{2\pi}\int_{-\infty}^{\infty} \ln|P(s)| d\omega    
		 = \begin{cases}
			\sum_{k=1}^{\bar{m}_k}  \bar{r}_k + \sum_{i=1}^{\bar{n}_x}  \bar{p}_i - \sum_{i=1}^{\bar{m}_x}  \bar{z}_i, 	\hspace{68pt} {\rm if \ }m_x+n >  n_x + m + 1;\\
			\sum_{k=1}^{\bar{m}_k} \bar{r}_k + \sum_{i=1}^{\bar{n}_x} \bar{p}^{}_i - \sum_{i=1}^{\bar{m}_x}\bar{z}_i  - \frac{K}{2K_x},		\hspace{36pt} {\rm if \ }m_x+n = n_x + m + 1;\\
			\sum_{k=1}^{\bar{m}_k} \bar{r}_k + \sum_{i=1}^{m_x - \bar{m}_x} z_i + \sum_{i=1}^{\bar{n}_x} \bar{p}_i - \sum_{i=1}^{n_x - \bar{n}_x}p_i - \sum_{j=1}^{m}z_j + \sum_{j=1}^{n - \bar{n}_x}p_j,\\
			\hspace{165pt} {\rm if \ }m_x + n = n_x + m  {\rm \ and \ } K = 2K_x; \\
			\pm \infty, \hspace{267pt} {\rm otherwise;}
		\end{cases}
	\end{align}
	where $\{\bar{r}_k\}_{k=1}^{\bar{m}_k}$ are the NMP zeros of $G_x(s)-F(s)G_y(s)$ excluding $\{\bar{p}_i\}_{i=1}^{\bar{n}_x}$.
\end{theorem}
\begin{proof}[\textbf{Proof}]
	Given the filtering sensitivity function $P(s) = [G_x(s) - F(s) G_y(s)]G_x^{-1}(s)$ in \eqref{p_fun}, the continuous-time filtering $P$-integral satisfies
	\begin{equation}\label{thm1_eq1}
		\int_{-\infty}^\infty \ln|P(s)|d\omega  = \int_{-\infty}^\infty \left(  \ln|G_x - FG_y| + \ln |G_x^{-1}|  \right)  d\omega.
	\end{equation}
	Substituting (\ref{TF_Gx}-\ref{TF_FG}) into \eqref{thm1_eq1}, we have $G_x(s) - F(s)G_y(s)    = \Gamma_c(s) / \Lambda_c(s)$, where 
	\begin{subequations}
		\begin{align}
			\Gamma_c(s)  & = K_x  \prod_{i=1}^{m_x}(s-z_i) \prod_{j=1}^{n - \bar{n}_x}(s - p_j^{})   - K   \prod_{i=1}^{n_x - \bar{n}_x}(s-p^{}_{i})  \prod_{j=1}^{m}(s - z_j) ,  \label{thm1_eq2}  \\
			\Lambda_c(s) & =  \prod_{i=1}^{{\bar{n}}_{x}}(s - \bar{p}_{i}^{})  \prod_{i=1}^{n_x - \bar{n}_x}(s-p^{}_{i})  \prod_{j=1}^{n - \bar{n}_x}(s - p_j^{}). \label{thm1_eq3}
		\end{align}		
	\end{subequations}
	Since $F(s)G_y(s)G_x^{-1}(s)$ is proper by \hyperref[ass3]{A3)}, we have $m_x + n \geq n_x + m$, which induces the following three different scenarios in computing the continuous-time $P$-integral. 
	
	\vspace{0.5em}
	
	\noindent {\underline{Case 1}}: $m_x + n > n_x + m + 1$, which is the most general scenario. After we combine and factorize the polynomials in \eqref{thm1_eq2}, $\Gamma_c(s)$ can be reformulated as
	\begin{equation}\label{thm1_eq4}
		\Gamma_c(s) = K_x \prod_{i=1}^{\bar{n}_x}(s-\bar{p}_i^{}) \prod_{k=1}^{m_x +n-2\bar{n}_x}(s-r_k)
	\end{equation}
	where the unstable poles $\{ \bar{p}_i \}_{i=1}^{\bar{n}_x}$, shared by $G_x$ and $G_y$, must be canceled by the NMP zeros of $G_x - FG_y$ and thus appear in $\Gamma_c(s)$ by \hyperref[ass4]{A4)} and \hyperref[rem1]{Remark~1}, and $\{r_k\}_{k=1}^{m_x+n - 2\bar{n}_x}$ then denote the remaining zeros of $G_x - FG_y$ excluding $\{ \bar{p}_i \}_{i=1}^{\bar{n}_x}$. Equating the coefficients of $s^{m_x + n - \bar{n}_x - 1}$ in \eqref{thm1_eq2} and \eqref{thm1_eq4} gives
	\begin{equation}\label{thm1_eq5}
		\sum_{i=1}^{m_x}z_i + \sum_{j=1}^{n-\bar{n}_x} p^{}_j = \sum_{i=1}^{\bar{n}_x} \bar{p}_i + \sum_{k=1}^{m_x +n-2\bar{n}_x}r_k.
	\end{equation}
	Plugging \eqref{TF_Gx}, \eqref{thm1_eq3} and \eqref{thm1_eq4} into the sensitivity function $P(s) = [G_x(s) - F(s)G_y(s)]G_x^{-1}(s)$, we have
	\begin{equation}\label{thm1_eq6}
		P(s) = \frac{ \prod_{i=1}^{\bar{n}_x}(s-\bar{p}_i^{}) \prod_{k=1}^{m_x+n-2\bar{n}_x}(s-r_k) }{\prod_{i=1}^{m_x}(s-z_i)\prod_{j=1}^{n-\bar{n}_x}(s-p_j^{})} .
	\end{equation}
	Integrating the logarithm of $|P(s)|$ over the frequency range $\omega \in (-\infty, \infty)$ and applying \hyperref[lem2]{Lemma~2} to the integral yield
	\begin{align}\label{thm1_eq7}
		\hspace{-5pt}	\frac{1}{2\pi}\int_{-\infty}^{\infty}\ln |P(s)| d\omega =& \frac{1}{2} \cdot \Bigg[ \sum_{i=1}^{\bar{n}_x}|{\rm Re} \ \bar{p}^{}_i| + \sum_{k=1}^{m_x+n-2\bar{n}_x}|{\rm Re} \ r_k| - \sum_{i=1}^{m_x}|{\rm Re} \ z_i| - \sum_{j=1}^{n - \bar{n}_x}|{\rm Re} \ p^{}_j|  \Bigg].
	\end{align}
	We then look into the sums on the RHS of \eqref{thm1_eq7}. Since the real parts of $\{\bar{p}_i\}_{i=1}^{\bar{n}_x}$ are positive in \eqref{TF_Gx}, while the real parts of $\{p_j\}_{j=1}^{n-\bar{n}_x}$ are negative in (\ref{TF_Gy}-\ref{TF_FG}), the first and fourth sums in \eqref{thm1_eq7} satisfy $\sum_{i=1}^{\bar{n}_x}|{\rm Re} \ \bar{p}^{}_i| = \sum_{i=1}^{\bar{n}_x} {\rm Re} \ \bar{p}^{}_i$, and $\sum_{j=1}^{n - \bar{n}_x}|{\rm Re} \ p^{}_j| = -\sum_{j=1}^{n - \bar{n}_x}{\rm Re} \ p^{}_j$. To simplify the second sum in \eqref{thm1_eq7}, we decompose the sum $\sum_{k=1}^{m_x +n-2\bar{n}_x}r^{}_k = \sum_{k=1}^{\bar{m}_k}\bar{r}_k + \sum_{k=1}^{m_x +n-2\bar{n}_x - \bar{m}_k}r^{}_k $, where $\{\bar{r}_k\}_{k=1}^{\bar{m}_k}$ denotes the NMP zeros of $G_x-FG_y$ excluding $\{\bar{p}_i\}_{i=1}^{\bar{n}_x}$; substituting this decomposition into \eqref{thm1_eq5}, we obtain the sum of MP zeros in $G_x - FG_y$, $\sum_{k=1}^{m_x +n-2\bar{n}_x - \bar{m}_k}r^{}_k =  \sum_{i=1}^{\bar{m}_x}\bar{z}^{}_i + \sum_{i=1}^{m_x - \bar{m}_x}z_i^{} + \sum_{j=1}^{n-\bar{n}_x} p^{}_j - \sum_{i=1}^{\bar{n}_x} \bar{p}^{}_i - \sum_{k=1}^{\bar{m}_k} \bar{r}_k$. Hence, the second sum in \eqref{thm1_eq7} satisfies $\sum_{k=1}^{m_x+ n -2\bar{n}_x}|{\rm Re} \ r_k| = \sum_{k=1}^{\bar{m}_k} {\rm Re \ } \bar{r}_k^{} - \sum_{k=1}^{m_x+ n -2\bar{n}_x - \bar{m}_k} {\rm Re} \ r_k= 2 \cdot \sum_{k=1}^{\bar{m}_k} {\rm Re} \  \bar{r}_k^{} - \sum_{i=1}^{\bar{m}_x}  {\rm Re \ } \bar{z}_i - \sum_{i=1}^{m_x - \bar{m}_x}{\rm Re \ }z_i - \sum_{j=1}^{n - \bar{n}_x} {\rm Re \ } p_j^{} + \sum_{i=1}^{\bar{n}_x} {\rm Re \ }\bar{p}_i^{} $. Since the real parts of $\{\bar{z}_i\}_{i=1}^{\bar{m}_x}$ are positive, while the real parts of $\{z_i\}_{i=1}^{m_x - \bar{m}_x}$ are negative, the third sum in \eqref{thm1_eq7} can be simplified as $\sum_{i=1}^{m_x}|{\rm Re} \ z_i| = \sum_{i=1}^{\bar{m}_x}{\rm Re} \ \bar{z}_i - \sum_{i=1}^{m_x - \bar{m}_x}{\rm Re} \ z_i$. Substituting the preceding simplifications back into \eqref{thm1_eq7}, when $m_x + n > n_x + m + 1$, the continuous-time filtering $P$-integral satisfies
	\begin{equation}\label{thm1_eq8}
		\frac{1}{2\pi}\int_{-\infty}^{\infty}\ln |P(s)| d\omega = \sum_{k=1}^{\bar{m}_k}{\rm Re} \ \bar{r}_k + \sum_{i=1}^{\bar{n}_x} {\rm Re} \ \bar{p}_i - \sum_{i=1}^{\bar{m}_x}{\rm Re} \ \bar{z}_i.
	\end{equation}

	\noindent \underline{Case 2}: $m_x + n = n_x + m + 1$. In this case, after combining and factorizing the polynomials in \eqref{thm1_eq2}, we obtain the same expression for $\Gamma_c(s)$ as in \eqref{thm1_eq4}. Equating the coefficients of $s^{m_x + n - \bar{n}_x - 1}$ or $s^{n_x + m - \bar{n}_x}$ in \eqref{thm1_eq2} and \eqref{thm1_eq4} gives
	\begin{equation}\label{thm1_eq9}
		\sum_{i=1}^{m_x} z_i + \sum_{j=1}^{n - \bar{n}_x}p_j + \frac{K}{K_x} =  \sum_{i=1}^{\bar{n}_x} \bar{p}_i + \sum_{k=1}^{m_x + n - 2\bar{n}_x}r_k .
	\end{equation}
	The expressions of filtering the sensitivity function $P(s)$ and filtering $P$-integral are the same as \eqref{thm1_eq6} and \eqref{thm1_eq7}, and the first, third, and fourth sums on the RHS of \eqref{thm1_eq7} can be expanded and simplified by the same procedures as in Case 1. The only difference between the Cases 1 and 2 lies in the second sum in \eqref{thm1_eq7}. By \eqref{thm1_eq9}, the sum of MP zeros of $G_x - FG_y$ is $\sum_{k=1}^{m_x + n - 2\bar{n}_x - \bar{m}_k} r_k = \sum_{i=1}^{\bar{m}_x} \bar{z}_i + \sum_{i=1}^{m_x - \bar{m}_x} z_i + \sum_{j=1}^{n - \bar{n}_x}p_j + {K}/{K_x} - \sum_{i=1}^{\bar{n}_x} \bar{p}_i -  \sum_{k=1}^{\bar{m}_k} \bar{r}_k$. Hence, when $m_x + n = n_x + m + 1$, the second sum in \eqref{thm1_eq7} satisfies $\sum_{k=1}^{m_x + n - 2\bar{n}_x}|{\rm Re} \ r_k| = \sum_{k=1}^{\bar{m}_k}{\rm Re \ } \bar{r}_k - \sum_{k=1}^{m_x + n - 2\bar{n}_x - \bar{m}_k}{\rm Re \ }r_k =  2 \cdot \sum_{k=1}^{\bar{m}_k}{\rm Re \ }\bar{r}_k - \sum_{i=1}^{\bar{m}_x} {\rm Re \ }\bar{z}_i - \sum_{i=1}^{m_x - \bar{m}_x} {\rm Re \ }z_i - \sum_{j=1}^{n - \bar{n}_x}{\rm Re \ }p_j - {K}/{K_x} + \sum_{i=1}^{\bar{n}_x} \bar{p}_i$. Plugging the preceding simplifications back into \eqref{thm1_eq7}, when $m_x + n = n_x + m + 1$, the continuous-time filtering $P$-integral satisfies
	\begin{align}\label{thm1_eq10}
		\frac{1}{2\pi}\int_{-\infty}^{\infty}\ln |P(s)| d\omega =  \sum_{k=1}^{\bar{m}_k}{\rm Re \ } \bar{r}_k &+ \sum_{i=1}^{\bar{n}_x}{\rm Re \ }\bar{p}_i  - \sum_{i=1}^{\bar{m}_x} {\rm Re \ }\bar{z}_i- \frac{K}{2K_x}. \nonumber
	\end{align}

	\noindent \underline{Case 3}: $m_x + n = n_x + m$. In this scenario, after we combine and factorize the polynomials in \eqref{thm1_eq2}, $\Gamma_c(s)$ can be rewritten as 
	\begin{equation}\label{thm1_eq11}
		\Gamma_c(s)  = (K_x - K) \cdot \prod_{i=1}^{\bar{n}_x}(s - \bar{p}_i)\prod_{k=1}^{m_x + n - 2\bar{n}_x}(s - r_k).
	\end{equation}
	Equating the coefficients of $s^{m_x + n - \bar{n}_x - 1}$ or $s^{n_x + m - \bar{n}_x - 1}$ in \eqref{thm1_eq2} and \eqref{thm1_eq11} gives
	\begin{align}\label{thm1_eq12}
		 K_x \Bigg[ \sum_{i=1}^{m_x}z_i + \sum_{j=1}^{n - \bar{n}_x}p_j  \Bigg]  - K \Bigg[  \sum_{i=1}^{n_x - \bar{n}_x}p_i + \sum_{j=1}^{m}z_j  \Bigg]  =  (K_x - K) \Bigg[  \sum_{i=1}^{\bar{n}_x}\bar{p}_i +  \sum_{k=1}^{m_x + n - 2\bar{n}_x}r_k \Bigg]  
	\end{align}
	Plugging \eqref{TF_Gx}, \eqref{thm1_eq3} and \eqref{thm1_eq11} into \eqref{p_fun}, the filtering sensitivity function $P(s)$ can then be rewritten as
	\begin{equation*}\label{thm1_eq13}
		P(s) = \frac{(K_x - K)  \prod_{i=1}^{\bar{n}_x}(s-\bar{p}_i^{}) \prod_{k=1}^{m_x+n-2\bar{n}_x}(s-r_k)  }{K_x  \prod_{i=1}^{m_x}(s-z_i)  \prod_{j=1}^{n-\bar{n}_x}(s-p_j^{})}.
	\end{equation*}
	Integrating the logarithm of $|P(s)|$ over the entire frequency range and applying \hyperref[lem2]{Lemma~2} to the integral, the continuous-time $P$-integral then satisfies
	\begin{align}\label{thm1_eq14}
		&\frac{1}{2\pi}\int_{-\infty}^{\infty}\ln |P(s)| d\omega = \frac{1}{2} \cdot \Bigg[  \sum_{i=1}^{\bar{n}_x}|{\rm Re} \ \bar{p}^{}_i|+ \sum_{k=1}^{m_x+n-2\bar{n}_x}|{\rm Re} \ r_k|\nonumber \\
		& \hspace{125pt}  - \sum_{i=1}^{m_x}|{\rm Re} \ z_i| - \sum_{j=1}^{n - \bar{n}_x}|{\rm Re} \ p^{}_j| \Bigg]  + \frac{1}{2\pi} \int_{-\infty}^{\infty} \ln \left| \frac{K_x - K}{K_x} \right| d\omega.
	\end{align}
	Equation \eqref{thm1_eq14} shows that when $m_x + n = n_x + m$, the continuous-time $P$-integral is bounded if and only if $\ln |(K_x - K) / K_x| =  0$, which requires that $K = 2K_x$ or $K = 0$; otherwise, the $P$-integral is unbounded, i.e., $\pm\infty$. Since $K = 0$ is a trivial scenario, namely $G_y(s) = 0$ or $F(s) = 0$,  in the following we only discuss the case when $K = 2K_x$. The first, third, and fourth sums on the RHS of \eqref{thm1_eq14} can be expanded and simplified by the same procedures as in the Cases 1 and 2. We then consider the second sum in \eqref{thm1_eq14}. Substituting the decomposition $\sum_{k=1}^{m_x +n-2\bar{n}_x}r^{}_k = \sum_{k=1}^{\bar{m}_k}\bar{r}_k + \sum_{k=1}^{m_x +n-2\bar{n}_x - \bar{m}_k}r^{}_k $ into \eqref{thm1_eq12}, the sum of MP zeros of $G_x - FG_y$ then becomes $\sum_{k=1}^{m_x+n-2\bar{n}_x - \bar{m}_k}r_k = K_x / (K_x - K) \cdot [ \sum_{i=1}^{m_x}z_i + \sum_{j=1}^{n - \bar{n}_x}p_j  ]  - K / (K_x - K) \cdot [\sum_{i=1}^{n_x - \bar{n}_x}p_i + \sum_{j=1}^{m}z_j] - \sum_{i=1}^{\bar{n}_x} \bar{p}_i - \sum_{k=1}^{\bar{m}_k} \bar{r}_k$. Hence, the second sum in \eqref{thm1_eq14} can be expanded as $\sum_{k=1}^{m_x+n-2\bar{n}_x}|{\rm Re} \ r_k|=\sum_{k=1}^{\bar{m}_k}{\rm Re} \ \bar{r}_k - \sum_{k=1}^{m_x+n-2\bar{n}_x - \bar{m}_k}{\rm Re \ }r_k = 2 \cdot [\sum_{k=1}^{\bar{m}_k}{\rm Re \ }\bar{r}_k - \sum_{i=1}^{n_x - \bar{n}_x} {\rm Re \ }p_i - \sum_{j=1}^{m} {\rm Re \ }z_j]  + \sum_{i=1}^{\bar{m}_x} {\rm Re \ } \bar{z}_i + \sum_{i=1}^{m_x - \bar{m}_x}  {\rm Re \ }  z_i   + \sum_{j=1}^{n - \bar{n}_x}{\rm Re \ } p_j + \sum_{i=1}^{\bar{n}_x}{\rm Re \ }\bar{p}_i$. Substituting the preceding simplifications back into \eqref{thm1_eq14}, when $m_x + n = n_x + m$ and $K = 2K_x$, the continuous-time filtering $P$-integral satisfies
	\begin{align}\label{thm1_eq15}
		\hspace{-25pt}& \frac{1}{2\pi}\int_{-\infty}^{\infty}\ln |P(s)| d\omega = \sum_{k=1}^{\bar{m}_k} {\rm Re \ } \bar{r}_k + \sum_{i=1}^{m_x - \bar{m}_x} {\rm Re \ }z_i  + \sum_{i=1}^{\bar{n}_x}{\rm Re \ }\bar{p}_i \\
		& \hspace{163pt}  -\sum_{i=1}^{n_x - \bar{n}_x}{\rm Re \ }p_i - \sum_{j=1}^{m}{\rm Re \ }z_j + \sum_{j=1}^{n - \bar{n}_x}{\rm Re \ }p_j. \nonumber
	\end{align}
	Since complex zeros and poles always appear as conjugate pairs, summarizing \eqref{thm1_eq8}, \eqref{thm1_eq10} and \eqref{thm1_eq15} gives \eqref{thm1_eq0} in \hyperref[thm1]{Theorem~1}. This completes the proof. 	
\end{proof}
\begin{remark}
	Compared with \hyperref[lem1]{Lemma~1}, \hyperref[thm1]{Theorem~1} provides more explicit and detailed information on continuous-time filtering $P$-integral. \hyperref[thm1]{Theorem~1} shows that apart from the NMP zero sets $\mathcal{Z}_P \supset \{\bar{p}_i\}_{i=1}^{\bar{n}_x}$ and $\mathcal{Z}_x \supset \{\bar{z}_i\}_{i=1}^{\bar{m}_x}$ defined in \hyperref[lem1]{Lemma~1}, the value and boundedness of the continuous-time $P$-integral are also determined by the leading coefficients $K$ and $K_x$, MP zeros $\{z_i\}_{i=1}^{m_x - \bar{m}_x}$ and $\{ z_j \}_{j=1}^m$,  stable poles $\{p_i\}_{i=1}^{n_x - \bar{n}_x}$ and $\{p_j\}_{j=1}^{n - \bar{n}_x}$, and relative degrees of transfer functions $G_x$ and $FG_y$. Meanwhile, when $m_x+n = n_x + m$, $\mathcal{Z}_x$ will not impact the value of the $P$-integral. The boundedness conditions given in \hyperref[thm1]{Theorem~1} also help the designers to better understand the scope and limitations of continuous-time $P$-integral.
\end{remark}

\subsection{Continuous-Time Filtering $M$-Integral}
We then consider the continuous-time $M$-integral. Since $\lim_{s\rightarrow\infty} M(s) = 0$, which can be verified by substituting (\ref{TF_Gx}-\ref{TF_F}) into the sensitivity functions \eqref{pm_funs} and complementary constraint \eqref{comp_constraint}, a weighting function $1/\omega^2$ is introduced to control the boundedness of continuous-time $M$-integral. This weighting function or weighted integral was also adopted in \hyperref[lem1]{Lemma~1} and \cite{Middleton_1991, Wan_TAC_2020}, wherein the authors investigated the complementary sensitivity Bode integral of continuous-time LTI control systems. The following theorem presents the result of continuous-time filtering $M$-integral derived from the simplified analysis. 
\begin{theorem}\label{thm2}
	For the continuous-time LTI filtering systems subject to plant models \eqref{TF_Gx}, \eqref{TF_Gy} and filter \eqref{TF_F}, as \hyperref[fig1]{Fig.~1} shows, the Bode-type integral of the filtering sensitivity function $M(s)$ defined in \eqref{m_fun} satisfies
	\begin{align}\label{thm2_eq0}
		\frac{1}{2\pi}\int_{-\infty}^{\infty} \ln |M(s)| \ \frac{d\omega}{\omega^2}= \begin{cases}
			\frac{1}{2} \cdot \big[\sum_{j=1}^{n - \bar{n}_x}  {p}^{-1}_j -\sum_{i=1}^{n_x - \bar{n}_x} {p}^{-1}_i  + \sum_{j=1}^{m}|{\rm Re} \ z^{-1}_j|   - \sum_{i=1}^{m_x}|{\rm Re} \  z^{-1}_i| \big], \\
			\hspace{82pt} {\rm if } \ |  K  \prod_{i=1}^{n_x - \bar{n}_x}{p}_i   \prod_{j=1}^{m}z_j   | = |  K_x\prod_{i=1}^{m_x}z_i \prod_{j=1}^{n - \bar{n}_x}{p}_j    |;\\
			\pm \infty, \hspace{252pt} {\rm otherwise}. 
		\end{cases}
	\end{align}
\end{theorem}
\begin{proof}[\textbf{Proof}]
	Plugging (\ref{TF_Gx})-(\ref{TF_FG}) into the continuous-time filtering sensitivity function $M(s)$ in \eqref{m_fun} gives
	\begin{equation}\label{thm2_eq1}
		M(s)=  \frac{K \cdot \prod_{i=1}^{n_x - \bar{n}_x}(s - p_i)\prod_{j=1}^{m}(s - z_j)}{K_x \cdot \prod_{i=1}^{m_x}(s - z_i) \prod_{j=1}^{n - \bar{n}_x}(s - p_j)  }  .
	\end{equation}
	To evaluate the weighted $M$-integral in \eqref{thm2_eq0}, consider the frequency transformation $\omega = -\tilde{\omega}^{-1}$, such that $s = j\omega = (j \tilde{\omega})^{-1} = \tilde{s}^{-1}$ and $M(s) = M(\tilde{s}^{-1}) = \tilde{M}(\tilde{s})$. By applying this frequency transformation to the continuous-time $M$-integral in \eqref{thm2_eq0}, we have 
	\begin{equation}\label{thm2_eq2}
		\int_{-\infty}^{\infty} \ln |M(s)| \frac{d\omega}{\omega^2} = \int_{-\infty}^{\infty} \ln |\tilde{M}(\tilde{s})| d\tilde{\omega},
	\end{equation}
	where the weighting function $1/\omega^2$ is eliminated on the RHS of \eqref{thm2_eq2}. Substituting $s = \tilde{s}^{-1}$ into \eqref{thm2_eq1} and multiplying the numerator and denominator by $\tilde{s}^{n+m_x - \bar{n}_x}$, we obtain $\tilde{M}(\tilde{s})$ in \eqref{thm2_eq2}
	\begin{align}\label{thm2_eq3}
		\tilde{M}(\tilde{s}) & = (-1)^{-\varepsilon} \cdot \frac{K  \prod_{i=1}^{n_x - \bar{n}_x}{p}_i    \prod_{j=1}^{m}z_j  }{K_x  \prod_{i=1}^{m_x}z_i  \prod_{j=1}^{n - \bar{n}_x}{p}_j } \cdot \hat{M}(\tilde{s}) ,
	\end{align}
	where $\varepsilon =  m_x +n - n_x - m \geq 0$, and
	\begin{equation*}\label{thm2_eq4}
		\hat{M}(\tilde{s}) = \frac{\tilde{s}^{\varepsilon} \cdot \prod_{i=1}^{n_x - \bar{n}_x}(\tilde{s} -  {p}^{-1}_i) \prod_{j=1}^{m}(\tilde{s}- z^{-1}_j ) }{\prod_{i=1}^{m_x}(\tilde{s} - z^{-1}_i)   \prod_{j=1}^{n-\bar{n}_x}(\tilde{s} - {p}^{-1}_j)}.
	\end{equation*}
	Plugging $\tilde{M}(\tilde{s})$ or \eqref{thm2_eq3} back into the RHS integral of \eqref{thm2_eq2}, the continuous-time $M$-integral then satisfies
	\begin{align}\label{thm2_eq5}
		\int_{-\infty}^{\infty} \ln|\tilde{M}(\tilde{s})| d\tilde{\omega} & = \int_{-\infty}^{\infty} \ln | \hat{M}(\tilde{s}) | d\tilde{\omega} +\int_{-\infty}^{\infty} \ln \left| \frac{K \prod_{i=1}^{n_x - \bar{n}_x}{p}_i  \prod_{j=1}^{m}z_j  }{K_x  \prod_{i=1}^{m_x}z_i  \prod_{j=1}^{n - \bar{n}_x}{p}_j }  \right| d\tilde{\omega}, \nonumber
	\end{align}
	which is bounded if and only if
	\begin{equation}\label{thm2_eq6}
		\left|  K  \prod_{i=1}^{n_x - \bar{n}_x}{p}_i \prod_{j=1}^{m}z_j \right| = \left|  K_x  \prod_{i=1}^{m_x}z_i   \prod_{j=1}^{n - \bar{n}_x}{p}_j    \right|.
	\end{equation}
	When condition \eqref{thm2_eq6} is violated, the $M$-integral in \eqref{thm2_eq0} and the integrals in \eqref{thm2_eq2} are unbounded, i.e., $\pm \infty$. When the condition \eqref{thm2_eq6} is fulfilled, by applying \hyperref[lem2]{Lemma~2} to the integral of $\hat{M}(\tilde{s})$ in \eqref{thm2_eq5}, we obtain the continuous-time filtering $M$-integral
	\begin{align}\label{thm2_eq7}
		\hspace{-5pt}\frac{1}{2\pi}\int_{-\infty}^{\infty} \ln |M(s)| \ \frac{d\omega}{\omega^2} = & \frac{1}{2} \cdot \bigg[ \sum_{j=1}^{n - \bar{n}_x} \textrm{Re} \ {p}^{-1}_j  - \sum_{i=1}^{n_x - \bar{n}_x} \textrm{Re} \ {p}^{-1}_i + \sum_{j=1}^{m}|\textrm{Re} \ z^{-1}_j|  - \sum_{i=1}^{m_x}|\textrm{Re} \  z^{-1}_i|  \bigg],
	\end{align}
	which implies \eqref{thm1_eq0} in \hyperref[thm2]{Theorem~2} with the fact that complex poles always appear as conjugate pairs. Meanwhile, by letting $\{\bar{z}_j\}_{j=1}^{\bar{m}}$ be the NMP zeros of $F(s)G_y(s)$ in \eqref{TF_FG}, or $\{{z}_j\}_{j=1}^{{m}}  = \{\bar{z}_j\}_{j=1}^{\bar{m}} \cup \{{z}_j\}_{j=1}^{{m - \bar{m}}}$, we can further expand or simplify the last two sums in \eqref{thm2_eq0} and \eqref{thm2_eq7} as $\sum_{i=1}^{m_x}|{\rm Re} \ z_i^{-1}| = \sum_{i=1}^{\bar{m}_x} \bar{z}_i^{-1} - \sum_{i=1}^{m_x - \bar{m}_x} z^{-1}_i$ and $\sum_{j=1}^{m}|\textrm{Re} \ z^{-1}_j| = \sum_{j=1}^{\bar{m}}\bar{z}_j^{-1} - \sum_{j=1}^{m - \bar{m}}z_j^{-1}$. This completes the proof. 
\end{proof}
\begin{remark}
	\hyperref[thm2]{Theorem~2} shows that in addition to the NMP zero sets $\mathcal{Z}_M \supset \{\bar{z}_j\}_{j=1}^{\bar{m}}$ and $\mathcal{Z}_x \supset \{\bar{z}_i\}_{i=1}^{\bar{m}_x}$ defined in \hyperref[lem1]{Lemma~1}, the value and boundedness of continuous-time filtering $M$-integral are also determined by the leading coefficients $K$ and $K_x$, MP zeros $\{z^{}_i\}_{i=1}^{m_x - \bar{m}_x}$ and $\{z^{}_j\}_{j=1}^{m - \bar{m}}$, and stable poles $\{p^{}_i\}_{i=1}^{n_x - \bar{n}_x}$ and $\{p^{}_j\}_{j=1}^{n - \bar{n}_x}$. More importantly, although keeping $|M(s)|$ small at high-frequency regime can be a good practice in filter design \cite{Goodwin_TAC_1997, Seron_2012}, this property also makes the continuous-time $M$-integral difficult to converge despite the aid of weighting function, which was not noticed in the previous papers. For a more distinguishable characterization of the sensitivity function $M(s)$, designers might resort to other types of trade-off integrals, such as the Poisson-type $M$-integral introduced in \cite{Goodwin_TAC_1997}. 
\end{remark}

\section{Discrete-Time Filtering Trade-offs}\label{sec4}
Discrete-time filtering $P$- and $M$-integrals, which have not been previously investigated, are analyzed in this section with the simplified method. Compared with the continuous-time trade-off integrals, since the frequency domain of discrete-time integrals is $\omega \in [-\pi, \pi]$, the boundedness of discrete-time trade-off integrals is more tractable to control.

\subsection{Discrete-Time Filtering $P$-Integral}
The simplified analysis result of discrete-time filtering $P$-integral is presented in the following theorem. 
\begin{theorem}\label{thm3}
	For the discrete-time LTI filtering system described by plant models \eqref{DTF_Gx}, \eqref{DTF_Gy} and filter \eqref{DTF_F}, as \hyperref[fig1]{Fig.~1} shows, the Bode-type integral of filtering sensitivity function $P(z)$ defined in \eqref{p_fun} satisfies
	\begin{align}\label{thm3_eq0}
		\frac{1}{2\pi} \int_{-\pi}^{\pi} \log |P(z)| d\omega  =\begin{cases}
			\sum_{i=1}^{\bar{n}_x} \log |\bar{p}_i| - \sum_{i=1}^{\bar{m}_x} \log |\bar{z}_i| + \sum_{k=1}^{\bar{n}_k} \log|\bar{r}_k|,\\
			\hspace{210pt} {\rm if} \ m_x+n \geq n_x+m + 1;\\
			\sum_{i=1}^{\bar{n}_x}\log|\bar{p}_i| - \sum_{i=1}^{\bar{m}_x}\log|\bar{z}_i|  + \sum_{k=1}^{\bar{n}_k}\log |\bar{r}_k| + \log |(K_x-K) / K_x|, \\
			\hspace{205pt}  \hspace{26pt} {\rm if} \ m_x + n = n_x + m;\\
		\end{cases}
	\end{align}
	where $\{\bar{r}_k\}_{k=1}^{\bar{n}_k}$ stand for the NMP zeros of $G_x(z) - FG_y(z)$ excluding $\{\bar{p}_i\}_{i=1}^{\bar{n}_x}$.
\end{theorem}
\begin{proof}[\textbf{Proof}]
	Plugging the filtering sensitivity function $P(z)$ or \eqref{p_fun} into the discrete-time $P$-integral defined in \eqref{thm3_eq0} gives
	\begin{equation}\label{thm3_eq1}
		\int_{-\pi}^{\pi} \log|P(z)| d\omega = \int_{-\pi}^{\pi}(\log|G_x - FG_y| + \log |G_x^{-1}|) d\omega.
	\end{equation}
	By substituting (\ref{DTF_Gx}-\ref{DTF_FG}) into \eqref{thm3_eq1}, we have $G_x(z) - F(z) \cdot G_y(z) = \Gamma_d(z) / \Lambda_d(z)$, where
	\begin{subequations}
		\begin{align}
			\Gamma_d(z) &= K_x \prod_{i=1}^{m_x}(z - z_i) \prod_{j=1}^{n - \bar{n}_x}(z - p_j) 	- K  \prod_{i=1}^{n_x- \bar{n}_x}(z - p_i)  \prod_{j=1}^{m}(z - z_j),    \label{thm3_eq2}\\
			\Lambda_d(z) & = \prod_{i=1}^{\bar{n}_x}(z - \bar{p}_i) \prod_{i=1}^{n_x - \bar{n}_x}(z-  {p}_i) \prod_{j=1}^{n - \bar{n}_x}(z - p_j). \label{thm3_eq3}
		\end{align}
	\end{subequations}
	Since $F(z)G_y(z)G_x^{-1}(z)$ is proper according to \hyperref[ass3]{A3)}, we have $m_x + n \geq n_x + m$, which categorizes the discrete-time $P$-integral into the following two scenarios. 
	
	\vspace{0.5em}
	
	\noindent {\underline{Case 1}}: $m_x+n \geq n_x + m + 1$. After we combine and factorize the polynomials in \eqref{thm3_eq2}, $\Gamma_d(z)$ becomes
	\begin{equation}\label{thm3_eq4}
		\Gamma_d(z) = K_x \cdot \prod_{i=1}^{\bar{n}_x}(z - \bar{p}_i) \prod_{k=1}^{m_x + n - 2\bar{n}_x}(z - r_k),
	\end{equation}
	where according to \hyperref[ass4]{A4)} and \hyperref[rem1]{Remark~1}, unstable poles $\{\bar{p}_i\}_{i=1}^{\bar{n}_x}$, shared by $G_x(z)$ and $G_y(z)$, must be canceled by the NMP zeros of $G_x(z) - F(z)G_y(z)$ and thus appear in $\Gamma_d(z)$, and $\{r_k\}_{k=1}^{m_x + n - 2\bar{n}_x}$ denote the remaining zeros. Substituting \eqref{DTF_Gx}, \eqref{thm3_eq3} and \eqref{thm3_eq4} into the filtering sensitivity function $P(z) = [G_x(z) - F(z)G_y(z)]G_x^{-1}(z)$ yields
	\begin{equation}\label{thm3_eq5}
		P(z) = \frac{ \prod_{i=1}^{\bar{n}_x}(z - \bar{p}_i)  \prod_{k=1}^{n+m_x - 2\bar{n}_x}(z - r_k)  }{ \prod_{i=1}^{m_x}(z - z_i)   \prod_{j=1}^{n - \bar{n}_x}(z - p_j)}. 
	\end{equation}
	Plugging \eqref{thm3_eq5} into the discrete-time $P$-integral defined in \eqref{thm3_eq0} and applying \hyperref[lem3]{Lemma~3} to the integral, we have 
	\begin{equation}\label{thm3_eq6}
		\frac{1}{2\pi} \int_{-\pi}^{\pi} \log |P(z)| d\omega = \sum_{i=1}^{\bar{n}_x} \log |\bar{p}_i|   - \sum_{i=1}^{\bar{m}_x} \log |\bar{z}_i| +  \sum_{k=1}^{\bar{n}_k} \log|\bar{r}_k|
	\end{equation}
	where $\{\bar{r}_k\}_{k=1}^{\bar{n}_k}$ denote the NMP zeros in $\{r_k\}_{k=1}^{m_x + n - 2\bar{n}_x}$, or equivalently, NMP zeros of $G_x - FG_y$ excluding $\{ \bar{p}_i \}_{i=1}^{\bar{n}_k}$.
	
	\noindent {\underline{Case 2}}: $m_x + n = n_x + m$. In this case, combining and factorizing the polynomials in \eqref{thm3_eq2} give
	\begin{equation}\label{thm3_eq7}
		\Gamma_d(z) = (K_x - K) \cdot \prod_{i=1}^{\bar{n}_x}(z - \bar{p}_i) \prod_{k=1}^{n+m_x - 2\bar{n}_x}(z - r_k),
	\end{equation}
	where $\{\bar{p}_i\}_{i=1}^{\bar{n}_x}$ and $\{r_k\}_{k=1}^{n+m_x - 2\bar{n}_x}$ follow the same definitions and reasoning as in \eqref{thm3_eq4}. Plugging \eqref{DTF_Gx}, \eqref{thm3_eq3} and \eqref{thm3_eq7} into $P(z) = [G_x(z) - F(z)G_y(z)] G_x^{-1}(z)$ or \eqref{p_fun}, the discrete-time filtering sensitivity function $P(z)$ becomes
	\begin{equation}\label{thm3_eq8}
		P(z) = \frac{(K_x - K) \cdot \prod_{i=1}^{\bar{n}_x}(z - \bar{p}_i)  \prod_{k=1}^{n+m_x -2 \bar{n}_x}(z- r_k) }{K_x \cdot \prod_{i=1}^{m_x}(z - z_i)   \prod_{j=1}^{n - \bar{n}_x}(z - p_j)}.
	\end{equation}
	By substituting \eqref{thm3_eq8} into the discrete-time $P$-integral defined in \eqref{thm3_eq0} and applying \hyperref[lem3]{Lemma~3} to the integral, we have
	\begin{align}\label{thm3_eq9}
		 \frac{1}{2\pi} \int_{-\pi}^{\pi} \log |M(z)| d\omega = \sum_{i=1}^{\bar{n}_x}\log|\bar{p}_i|  - \sum_{i=1}^{\bar{m}_x}\log|\bar{z}_i|  + \sum_{k=1}^{\bar{n}_k}\log |\bar{r}_k| +\log \left|  \frac{K_x-K}{K_x}\right|. 
	\end{align}
	Summarizing \eqref{thm3_eq6} and \eqref{thm3_eq9}, we obtain \eqref{thm3_eq0} in \hyperref[thm3]{Theorem~3}.
\end{proof}

\subsection{Discrete-Time Filtering $M$-Integral}
The following theorem presents the simplified analysis result of the discrete-time filtering $M$-integral.
\begin{theorem}\label{thm4}
	For the discrete-time LTI filtering system described by plant models \eqref{DTF_Gx}, \eqref{DTF_Gy} and filter \eqref{DTF_F}, as \hyperref[fig1]{Fig.~1} shows, the Bode-type integral of filtering sensitivity function $M(z)$ defined in \eqref{m_fun} satisfies
	\begin{align}\label{thm4_eq0}
		\frac{1}{2\pi} \int_{-\pi}^{\pi} \log |M(z)| d\omega = \sum_{j=1}^{\bar{m}}\log|\bar{z}_j| - \sum_{i=1}^{\bar{m}_x}\log|\bar{z}_i| + \log\left| \frac{K}{K_x}\right| ,
	\end{align}
	where $\{\bar{z}_j\}_{j=1}^{\bar{m}}$ denote the NMP zeros of $F(z)G_y(z)$. 
\end{theorem}
\begin{proof}[\textbf{Proof}]
	By substituting (\ref{DTF_Gx}-\ref{DTF_FG}) into the filtering sensitivity function $M(z)$ in \eqref{m_fun}, we have
	\begin{equation}\label{thm4_eq1}
		M(z) = \frac{K \cdot  \prod_{i=1}^{n_x - \bar{n}_x}(z - p_i) \prod_{j=1}^{m}(z - z_j) }{K_x \cdot  \prod_{i=1}^{m_x}(z - z_i)  \prod_{j=1}^{n - \bar{n}_x} (z - p_j) }. 
	\end{equation}
	Plugging \eqref{thm4_eq1} into the discrete-time $M$-integral in \eqref{thm4_eq0} and applying \hyperref[lem3]{Lemma~3} to the integral give \eqref{thm4_eq0} in \hyperref[thm4]{Theorem~4}. 
\end{proof}
\begin{remark}
	Similar to the results of continuous-time $P$- and $M$-integrals in \hyperref[thm1]{Theorems~1}, \hyperref[thm2]{2}, and \hyperref[lem1]{Lemma~1}, \hyperref[thm3]{Theorems~3} and~\hyperref[thm4]{4} show that the values of discrete-time $P$- and $M$-integrals are determined by the leading coefficients $K$, $K_x$, relative degrees of plants and filter, unstable poles of $G_x(z)$, as well as NMP zeros of $G_x(z)$, $F(z)G_y(z)$, and $G_x(z) -F(z)G_y(z)$. Meanwhile, with bounded frequency range, the discrete-time filtering trade-off integrals always converge despite the choices of plant models and filters.
\end{remark}

\section{Illustrative Examples}\label{sec5}
Illustrative examples are presented in this section to verify the validity and correctness of the simplified analysis results in \hyperref[thm1]{Theorems~1}-\hyperref[thm4]{4}. To verify \hyperref[thm1]{Theorem~1}, we choose the plant $G_x(s) =  1.67 \cdot (s+0.05) / (s+0.04)$ with $m_x = n_x = 1$. For the first case when $m_x + n > n_x + m + 1$, consider the plant $G_y(s) = 1.25\cdot (s+0.06)(s+0.08) / [(s-0.03)(s+0.01)(s+0.07)]$ and filter $F(s) = 1.34 \cdot (s-0.03)(s+0.09) / (s+0.68)^3$ with $m = 3$ and $n = 5$. Substituting $G_x(s)$, $F(s)$, and $G_y(s)$ into the continuous-time filtering $P$-integral in \eqref{thm1_eq0}, the numerical computation of $P$-integral gives $(1/2\pi) \cdot \int_{-\infty}^{\infty} \ln|P(s)| d\omega \approx 0.0101$, and by applying \hyperref[thm1]{Theorem~1}, the theoretical result of simplified analysis is $\sum_{k=1}^{\bar{m}_k}{\rm Re} \ \bar{r}_k + \sum_{i=1}^{\bar{n}_x} {\rm Re} \ \bar{p}_i - \sum_{i=1}^{\bar{m}_x}{\rm Re} \ \bar{z}_i = 0.0101$, which matches the numerical result. For the second case when $m_x + n = n_x + m + 1$, we consider the same $G_x(s)$, $G_y(s)$ as in the first case and filter $F(s) = 1.34 \cdot (s - 0.03)(s+0.075)(s+0.09)/(s+0.68)^3$ such that $m_x = n_x = 1$, $m = 4$ and $n = 5$. Numerical computation of the $P$-integral gives $-0.5015$, and the simplified analysis shows $\sum_{k=1}^{\bar{n}_k} {\rm Re \ }\bar{r}_k + \sum_{i=1}^{\bar{n}_z} {\rm Re \ }\bar{p}^{}_i - \sum_{i=1}^{\bar{m}_x}{\rm Re \ }\bar{z}_i  - {K}/{2K_x} = - 1.25 \times 1.34 / (2 \times 1.67) \approx -0.5015$. For the third case when $m_x + n = n_x + m$ and $K = 2K_x$, we keep $G_x(s)$ unaltered and let $G_y(s) = 1.25 \cdot (s+0.06)(s+0.08) / [(s-0.03)(s+0.01)]$ and $F(s) = 2.672 \cdot (s-0.03)(s+0.075)(s+0.09) / (s+0.68)^3$. Numerical computation of the $P$-integral gives $(1/2\pi) \cdot \int_{-\infty}^{\infty} \ln|P(s)| d\omega \approx 0.3991$, and the simplified analysis in \hyperref[thm1]{Theorem~1} also gives $\sum_{k=1}^{\bar{m}_k} \bar{r}_k + \sum_{i=1}^{m_x - \bar{m}_x} z_i + \sum_{i=1}^{\bar{n}_x} \bar{p}_i - \sum_{i=1}^{n_x - \bar{n}_x}p_i - \sum_{j=1}^{m}z_j + \sum_{j=1}^{n - \bar{n}_x}p_j \approx 0.3991$. For the last scenario when $m_x+n = n_x+m$ and $K \ne 2K_x$, we choose the same $G_x(s)$ and $G_y(s)$ as in the third case and the filter $F(s) = 2.872 \cdot (s-0.03)(s+0.075)(s+0.09) / (s+0.68)^3$. Both numerical computation and simplified analysis indicate that the $P$-integral is $+\infty$ or unbounded.

To verify \hyperref[thm2]{Theorem~2}, consider the plants $G_x(s) = 1.5 \cdot (s+0.05) / (s+0.025)$ and $G_y(s) = 1.25 \cdot (s+0.075)(s+0.75) / [(s-0.05)(s+0.01)]$. For the case when $|  K  \prod_{i=1}^{n_x - \bar{n}_x}{p}_i   \prod_{j=1}^{m}z_j   | = |  K_x\prod_{i=1}^{m_x}z_i \prod_{j=1}^{n - \bar{n}_x}{p}_j  |$, let  $F(s) = (32/15) \cdot (s-0.05)(s+0.1)(s+0.25) /  (s+0.5)^3$. Numerical computation of the $M$-integral gives $({1}/{2\pi}) \cdot \int_{-\infty}^{\infty} \log|M(s)|/\omega^2 \ d\omega \approx -28.6667$, and simplified analysis shows that $({1} / {2}) \cdot  [\sum_{j=1}^{n - \bar{n}_x}  {p}^{-1}_j -\sum_{i=1}^{n_x - \bar{n}_x} {p}^{-1}_i  +  \sum_{j=1}^{m}|{\rm Re} \ z^{-1}_j| - \sum_{i=1}^{m_x}|{\rm Re} \  z^{-1}_i| ] = -86 / 3 \approx -28.6667$. For the case when $|  K  \prod_{i=1}^{n_x - \bar{n}_x}{p}_i   \prod_{j=1}^{m}z_j   | \neq |  K_x  \prod_{i=1}^{m_x}z_i  \cdot \prod_{j=1}^{n - \bar{n}_x}{p}_j  |$, we only alter the filter to $F(s) = 2 \cdot (s-0.05)(s+0.1)(s+0.25) / (s+0.5)^3$, and both numerical computation and simplified analysis show that the $M$-integral is $-\infty$.

We then verify the result of discrete-time filtering $P$-integral in \hyperref[thm3]{Theorem~3}. For the first case when $m_x + n \geq n_x + m + 1$, consider the plant models $G_x(z) = 1.5 \cdot (z-0.05) / (z-0.1)$, $G_y(z) = 1.25 \cdot (z-0.075)(z-0.025) / [(z-1.25)(z-0.75)(z-0.01)]$, and filter $F(z) = 1.75 \cdot (z-1.25)(z-0.25) / (z-0.5)^3$ with $m_x=n_x = 1$, $m = 3$, and $n= 5$. Numerical computation of $P$-integral shows that $({1}/{2\pi}) \cdot \int_{-\pi}^{\pi} \log|P(z)| d\omega \approx 1.0512$, and simplified analysis gives that $\sum_{i=1}^{\bar{n}_x} \log |\bar{p}_i| - \sum_{i=1}^{\bar{m}_x} \log |\bar{z}_i| + \sum_{k=1}^{\bar{n}_k} \log|\bar{r}_k| \approx \log(2.0722) \approx 1.0512$. For the second case when $m_x + n = n_x + m$, we keep $G_x(z)$ unaltered and consider the plant $G_y(z) = 1.25\cdot (z-0.075)(z-0.025) / [(z-1.25)(z-0.01)]$ and filter $F(z) = 1.75 \cdot (z-1.25)(z-0.75)(z-0.25) / (z-0.5)^3$ such that $m_x = n_x = 1$ and $m = n = 4$. Numerical computation gives that $({1}/{2\pi}) \cdot \int_{-\pi}^{\pi} \log|P(z)| d\omega \approx -1.1255$, and simplified analysis implies that $\sum_{i=1}^{\bar{n}_x}\log|\bar{p}_i| - \sum_{i=1}^{\bar{m}_x}\log|\bar{z}_i|  + \sum_{k=1}^{\bar{n}_k}\log |\bar{r}_k| + \log |(K_x-K) / K_x| =  \log|(1.5 - 1.25 \times 1.75) / 1.5| = -1.1255$.

To verify the result of discrete-time filtering $M$-integral in \hyperref[thm4]{Theorem~4}, consider the plant models $G_x(z) = 1.5\cdot(z-0.05)/(z-0.1)$, $G_y(z) = 1.25 \cdot (z-0.075)(z-0.025) / [(z-1.25)(z-0.01)]$, and filter $F(z) = 1.75 \cdot (z-1.25)(z-0.75)(z-0.25) / (z-0.5)^3$ such that $m_x = n_x = 1$ and $m = n = 4$. Numerical computation of $M$-integral shows that $(1/2\pi) \cdot \int_{-\pi}^{\pi} \log |M(z)| d\omega \approx 0.5443$, and the simplified analysis shows that $\sum_{j=1}^{\bar{m}}\log|\bar{z}_j| - \sum_{i=1}^{\bar{m}_x}\log|\bar{z}_i| + \log| {K} / {K_x}| = \log|1.75 \times 1.25 / 1.5| \approx 0.5443$.


\section{Conclusion and Discussion}\label{sec6}
A comprehensive analysis of continuous- and discrete-time filtering sensitivity integrals was performed by relying on the simplified approach. With less restrictive assumptions, new insights and hidden factors that determine the filtering sensitivity trade-offs were disclosed. Illustrative examples were presented and verified the simplified analysis results. By extending the simplified method to non-rational transfer functions, e.g., transfer functions with time delay, we can also conduct a simplified analysis on the sensitivity trade-off integrals of smoothing and prediction problems. Meanwhile, generalizing \hyperref[lem1]{Lemma~1} to the discrete-time filtering systems with complex analysis tools is also an open problem.

\section*{Acknowledgment}
This work was partially supported by AFOSR and NSF. The authors would specially acknowledge the readers and staff on arXiv.org.

\bibliographystyle{IEEEtran}
\bibliography{ref}

\end{document}